\documentclass{article}%
\usepackage{amsmath}
\usepackage{amsfonts}
\usepackage{amssymb}
\usepackage{graphicx}%
\providecommand{\U}[1]{\protect\rule{.1in}{.1in}}
\newtheorem{theorem}{Theorem}

\newtheorem{definition}[theorem]{Definition}

\newtheorem{proposition}[theorem]{Proposition}
\newtheorem{remark}[theorem]{Remark}

\newenvironment{proof}[1][Proof]{\noindent\textbf{#1.} }{\ \rule{0.5em}{0.5em}}
\begin{document}

\title{Locally Inertial Reference Frames in Lorentzian and Riemann-Cartan Spacetimes}
\author{J. F. T. Giglio$^{(1)}$ and W. A. Rodrigues Jr.$^{(2)}$\\$^{(1)}$ FEA-CEUNSP. 13320-902 - Salto, SP Brazil.\\$^{(2)}$ Institute of Mathematics Statistics and Scientific Computation\\IMECC-UNICAMP \\13083950 Campinas, SP Brazil\\walrod@ime.unicamp.br or walrod@mpc.com.br}
\maketitle

\begin{abstract}
In this paper we scrutinize the concept of locally inertial reference frames
(\textbf{LIRF}) in Lorentzian and Riemann-Cartan spacetime structures. We
present rigorous mathematical definitions for those objects, something that
needs preliminary a clear mathematical distinction between the concepts of
observers, reference frames, naturally adapted coordinate functions to a given
reference frame and which properties may characterize an inertial reference
frame (if any) in\ the Lorentzian and Riemann-Cartan structures. We hope to
have clarified some\ eventual obscure issues associated to the concept of
\textbf{LIRF} appearing in the literature, in particular the relationship
between \textbf{LIRF}s in Lorentzian and Riemann-Cartan spacetimes and
Einstein's most happy though, i.e., the equivalence principle.

\end{abstract}

\section{Introduction}

In this note we investigate if it is possible to have in a general
Riemann-Cartan spacetime a locally inertial reference frame in an analogous
sense in which this concept is defined in a Lorentzian spacetime that models
possible gravitational fields in General Relativity.

To answer the above question which is affirmative in a well defined sense we
are going to recall the precise definitions of the following fundamental concepts:

(i) definition of a \emph{general reference frame} in Lorentzian and
Riemann-Cartan spacetimes;

(ii) definition of \emph{observers} in a Lorentzian or Riemannian spacetime;

(ii) classification of \emph{reference frames} in Lorentzian
spacetimes\footnote{The classification of reference frames will not be
presented in this paper. For the Lorentzian spacetime case, see
\cite{rodcap2007}.};

(iii) definition of an \emph{inertial reference frame} (\textbf{IRF}) in
Minkowski spacetime;

(iv) definition of a \emph{locally inertial\ reference frame }(\textbf{LIRF})
in Lorentzian and Riemann-Cartan spacetimes.

However, to be possible to present precise definitions of the concepts just
mentioned we need to recall some basic facts of differential geometry and fix
some notation. This will be done in Section 2.

\section{Lorentzian and Riemann-Cartan Spacetimes}

To start we introduce a Lorentzian manifold as a pair $\langle
M,\boldsymbol{g}\rangle$ where $M$ is a 4-dimensional manifold and
$\boldsymbol{g\in}\sec T_{0}^{2}M$ is a Lorentz metric of signature $(1,3)$.
We suppose that $\langle M,\boldsymbol{g}\rangle$ is orientable by a global
volume form $\boldsymbol{\tau}_{\boldsymbol{g}}\in\sec%
{\textstyle\bigwedge\nolimits^{4}}
T^{\ast}M$ and also time orientable by an equivalence relation here denoted
$\uparrow$. We next introduce on $M$ two metric compatible connections, namely
$\mathring{D}$ the Levi-Civita connection of $\boldsymbol{g}$ and $D$ a
general Riemann-Cartan connection.

\begin{definition}
We call the pentuple $\langle M,\boldsymbol{g},\mathring{D},\boldsymbol{\tau
}_{\boldsymbol{g}},\uparrow\rangle$ a Lorentzian spacetime and the pentuple
$\langle M,\boldsymbol{g},D,\boldsymbol{\tau}_{\boldsymbol{g}},\uparrow
\rangle$ a Riemann-Cartan spacetime.
\end{definition}

\begin{remark}
Minkowski spacetime structure is denoted by $\langle M\simeq\mathbb{R}%
^{4},\boldsymbol{\eta},\overset{m}{D},\boldsymbol{\tau}_{\boldsymbol{\eta}%
},\uparrow\rangle$.
\end{remark}

Let $U,V,W\subset M$ with $U\cap V\cap W\neq\varnothing$ and introduce the
local charts $(\varphi,U)$ and $(\psi,V)$ and \ $(\chi,W)$ with coordinate
functions $\langle\xi^{\mu}\rangle$, $\langle x^{\mu}\rangle$, $\langle
x^{\prime\mu}\rangle$ respectively. Recall to fix notation that, e.g., given
$p\in M$ and $\mathcal{V}\subset\mathbb{R}^{4}$ we have
\begin{equation}
\psi:V\rightarrow\mathcal{V},\text{ \ \ }\psi(p)=(x^{0}(p),x^{1}%
(p),x^{2}(p),x^{3}(p)). \label{1}%
\end{equation}

The coordinate chart $\psi$ determines a so-called coordinate basis for $TV$
denoted by $\langle e_{\mu}=\partial/\partial x^{\mu}\rangle$. We denoted by
$\langle\vartheta^{\mu}=dx^{\mu}\rangle$ a basis for $T^{\ast}V$ dual to
$\langle e_{\mu}=\partial/\partial x^{\mu}\rangle$, this meaning that
$\vartheta^{\mu}(e_{\nu})=\delta_{\nu}^{\mu}$.

We also write%
\begin{align}
\boldsymbol{g}  &  =g_{\mu\nu}\vartheta^{\mu}\otimes\vartheta^{\nu}=g^{\mu\nu
}\vartheta_{\mu}\otimes\vartheta_{\nu},\nonumber\\
g_{\mu\nu}  &  :=\boldsymbol{g}(e_{\mu},e_{\nu}),\text{ \ }g^{\mu\nu
}:=g(\vartheta^{\mu},\vartheta^{\nu}), \label{2}%
\end{align}
were we denoted by $\langle e^{\mu}\rangle$ the reciprocal basis of $\langle
e_{\mu}\rangle$, i.e., we have
\begin{equation}
\boldsymbol{g}(e^{\mu},e_{\nu})=\delta_{\nu}^{\mu}. \label{3}%
\end{equation}
Moreover, we denote by $g$ $\boldsymbol{\in}\sec T_{2}^{0}M$ the metric on the
cotangent bundle and write%
\begin{equation}
g=g^{\mu\nu}e_{\mu}\otimes e_{\nu}=g_{\mu\nu}e^{\mu}\otimes e^{\nu}.
\label{2a}%
\end{equation}
Of course, $g^{\mu\alpha}g_{\alpha\nu}=\delta_{\nu}^{\mu}$. Also, we denoted
by $\langle\vartheta_{\mu}\rangle$ the reciprocal basis of $\langle
\vartheta^{\mu}\rangle$, i.e., $g(\vartheta^{\mu},\vartheta_{\nu})=\delta
_{\nu}^{\mu}$.\medskip

A curve in $M$ is a mapping
\[
c:\mathbb{R\supset}I\rightarrow M,\text{ \ \ \ }\tau\mapsto c(\tau).
\]
As usual the \emph{tangent vector} \emph{field} to the curve $c$ is denoted by
$c_{\ast}$ or by $\frac{d}{d\tau}$ as more convenient. In the coordinate basis
$\langle e_{\mu}=\partial/\partial x^{\mu}\rangle$ we write%
\begin{equation}
c_{\ast}=\frac{d}{d\tau}=v^{\mu}(\tau)\left.  \partial/\partial x^{\mu
}\right\vert _{\gamma(\tau)} \label{4}%
\end{equation}
In particular we write when $c(0)=p_{o}$,%
\begin{equation}
\left.  c_{\ast}\right\vert _{\tau=o}=v^{\mu}\left.  e_{\mu}\right\vert
_{p_{o}}\in T_{p_{o}}M. \label{4b}%
\end{equation}

To understand the reason for that notation, first take into account that the
coordinate representation of $c$ are the set of functions $x^{\mu}\circ
c(\tau)$ that we denoted using a sloop notation simply by $x^{\mu}(\tau)$.

Now, consider a function $\boldsymbol{f}:V\rightarrow\mathbb{R}$ and denote by
$f=\boldsymbol{f\circ}\psi^{-1}:\mathcal{V}\rightarrow\mathbb{R}$ its
representation as functions of the coordinates $\langle x^{\mu}\rangle$.
Moreover, consider the composite function $\boldsymbol{f\circ}c$ and its
representative
\begin{equation}
f(x^{\mu}(\tau)) \label{5}%
\end{equation}

Then the value of the function $\frac{d}{d\tau}\boldsymbol{f\circ}c$ at
$c(\tau_{o})=p_{o}$ is
\begin{equation}
\left.  \frac{d}{d\tau}\boldsymbol{f\circ}c(\tau)\right\vert _{p_{o}}:=\left.
\frac{d}{d\tau}f(x^{\mu}(\tau))\right\vert _{\tau=0}=v^{\mu}\left.  \partial
f/\partial x^{\mu}\right\vert _{p_{o}}, \label{6}%
\end{equation}
with
\begin{equation}
v^{\mu}:=\left.  \frac{dx^{\mu}(\tau)}{d\tau}\right\vert _{\tau=0}. \label{7}%
\end{equation}

The metric structure permit to classify curves as \emph{timelike},
\emph{spacelike} and \emph{lightlike}. We have
\begin{equation}
\left\{
\begin{array}
[c]{ccc}%
\boldsymbol{g}(c_{\ast},c_{\ast})>0 & \forall\tau\in I & c\text{ is
timelike}\\
\boldsymbol{g}(c_{\ast},c_{\ast})<0 & \forall\tau\in I & c\text{ is
spacelike}\\
\boldsymbol{g}(c_{\ast},c_{\ast})=0 & \forall\tau\in I & c\text{ is
ligthhlike}%
\end{array}
\right.  \label{8}%
\end{equation}

For timelike curve $c:\tau\mapsto c(\tau)$, such that $\boldsymbol{g}(c_{\ast
},c_{\ast})=1$ the parameter $\tau$ is called the \emph{propertime}.

Given $U,V\subset M$ and coordinate functions $\langle\xi^{\mu}\rangle,$
$\langle x^{\mu}\rangle$ covering $U$ and $V$ for the structure $\langle
M,\boldsymbol{g},\mathring{D},\boldsymbol{\tau}_{\boldsymbol{g}}%
,\uparrow\rangle$ and coordinate functions $\langle\varsigma^{\mu}\rangle,$
$\langle x^{\mu}\rangle$ covering $U$ and $V$ for the structure $\langle
M,\boldsymbol{g},D,\boldsymbol{\tau}_{\boldsymbol{g}},\uparrow\rangle$ we fix
here the following notation%

\begin{align}
D_{e_{\mu}}\vartheta^{\nu}  &  :=-\Gamma_{\cdot\mu\alpha}^{\nu\cdot\cdot
}\vartheta^{\alpha},\text{ \ \ }D_{e_{\mu}}e_{\nu}:=\Gamma_{\cdot\mu\nu
}^{\alpha\cdot\cdot}e_{\alpha},\nonumber\\
D_{\boldsymbol{\partial/\partial\varsigma}^{\mu}}d\varsigma^{\nu}  &
:=-\mathbf{\Gamma}_{\cdot\mu\alpha}^{\nu\cdot\cdot}d\varsigma^{\alpha},\text{
\ \ }D_{_{\partial/\partial\varsigma^{\mu}}}\partial/\partial\varsigma^{\nu
}:=\mathbf{\Gamma}_{\cdot\mu\nu}^{\alpha\cdot\cdot}\partial/\partial
\varsigma^{\alpha},\nonumber\\
\mathring{D}_{e_{\mu}}\vartheta^{\nu}  &  :=-\mathring{\Gamma}_{\cdot\mu
\alpha}^{\nu\cdot\cdot}\vartheta^{\alpha},\text{ \ \ }\mathring{D}_{e_{\mu}%
}e_{\nu}:=\mathring{\Gamma}_{\cdot\mu\nu}^{\alpha\cdot\cdot}e_{\alpha
},\nonumber\\
\mathring{D}_{\boldsymbol{\partial/\partial\xi}^{\mu}}d\xi^{\nu}  &
:=-\boldsymbol{\mathring{\Gamma}}_{\cdot\mu\alpha}^{\nu\cdot\cdot}d\xi
^{\alpha},\text{ \ \ }\mathring{D}_{_{\boldsymbol{\partial/\partial\xi}^{\mu}%
}}\partial/\partial\boldsymbol{\xi}^{\nu}=\mathbf{\mathring{\Gamma}}_{\cdot
\mu\nu}^{\alpha\cdot\cdot}\partial/\partial\boldsymbol{\xi}^{\alpha}.
\label{Ds}%
\end{align}

For the connection coefficients in coordinate basis $\langle\partial/\partial
x^{\prime\mu}\rangle,\langle dx^{\prime\mu}\rangle$ we use\textbf{ }%
$\Gamma_{\cdot\mu\alpha}^{\prime\nu\cdot\cdot}$. Finally for an arbitrary
basis $\langle\boldsymbol{e}_{\mu}\rangle$ for $T(U\cap V\cap W)$ and dual
basis $\langle\boldsymbol{\theta}^{\mu}\rangle$ for $T^{\ast}(U\cap V\cap W)$
such that%
\begin{equation}
\lbrack\boldsymbol{e}_{\mu},\boldsymbol{e}_{\nu}]=\boldsymbol{c}_{\cdot\mu\nu
}^{\alpha\cdot\cdot}\boldsymbol{e}_{\alpha} \label{com}%
\end{equation}
we write%
\begin{align}
D_{\mathbf{e}_{\mu}}\boldsymbol{\theta}^{\nu}  &  :=-\boldsymbol{\gamma
}_{\cdot\mu\alpha}^{\nu\cdot\cdot}\boldsymbol{\theta}^{\alpha}%
\boldsymbol{,\ \ }D_{\mathbf{e}_{\mu}}\boldsymbol{e}_{\nu}:=\boldsymbol{\gamma
}_{\cdot\mu\nu}^{\alpha\cdot\cdot}\boldsymbol{e}_{\alpha},\nonumber\\
\mathring{D}_{\boldsymbol{e}_{\mu}}\boldsymbol{\theta}^{\nu}  &
:=-\boldsymbol{\mathring{\gamma}}_{\cdot\mu\alpha}^{\nu\cdot\cdot
}\boldsymbol{\theta}^{\alpha}\boldsymbol{,\ \ }\mathring{D}_{\boldsymbol{e}%
_{\mu}}\boldsymbol{e}_{\nu}:=\boldsymbol{\mathring{\gamma}}_{\cdot\mu\nu
}^{\alpha\cdot\cdot}\boldsymbol{e}_{\alpha}. \label{Dsb}%
\end{align}

\subsection{Relation between $\Gamma_{\cdot\mu\nu}^{\lambda\cdot\cdot}$ and
$\mathring{\Gamma}_{\cdot\mu\nu}^{\lambda\cdot\cdot}$}

We have that\footnote{See, e.g., \cite{gr2011}.}:%

\begin{equation}
\Gamma_{\cdot\mu\nu}^{\lambda\cdot\cdot}=\mathring{\Gamma}_{\cdot\mu\nu
}^{\lambda\cdot\cdot}+K_{\cdot\mu\nu}^{\lambda\cdot\cdot} \label{3a}%
\end{equation}
where
\begin{align}
K_{\cdot\mu\nu}^{\lambda\cdot\cdot}  &  :=\frac{1}{2}(T_{\cdot\mu\nu}%
^{\lambda\cdot\cdot}+S_{\cdot\mu\nu}^{\lambda\cdot\cdot})\nonumber\\
&  =\frac{1}{2}g^{\lambda\beta}g_{\beta\alpha}T_{\cdot\mu\nu}^{\alpha
\cdot\cdot}-\frac{1}{2}g^{\lambda\sigma}g_{\mu\alpha}T_{\cdot\nu\sigma
}^{\alpha\cdot\cdot}-\frac{1}{2}g^{\lambda\sigma}g_{\nu\alpha}T_{\cdot
\mu\sigma}^{\alpha\cdot\cdot}\text{ }\label{3b}\\
&  =\frac{1}{2}(T_{\cdot\mu\nu}^{\lambda\cdot\cdot\hspace{0.01in}}-T_{\mu
\nu\cdot}^{\cdot\cdot\lambda}+T_{\nu\cdot\mu}^{\cdot\lambda\cdot}).\nonumber
\end{align}
and%

\begin{align}
T_{\cdot\mu\nu}^{\lambda\cdot\cdot}  &  =\Gamma_{\cdot\mu\nu}^{\lambda
\cdot\cdot}-\Gamma_{\cdot\nu\mu}^{\lambda\cdot\cdot}=-T_{\cdot\nu\mu}%
^{\lambda\cdot\cdot},\label{3c}\\
S_{\cdot\mu\nu}^{\lambda\cdot\cdot}  &  =-g^{\lambda\sigma}(g_{\mu\alpha
}T_{\cdot\nu\sigma}^{\alpha\cdot\cdot}+g_{\nu\alpha}T_{\cdot\mu\sigma}%
^{\alpha\cdot\cdot})=S_{\cdot\nu\mu}^{\lambda\cdot\cdot}\text{ }. \label{3d}%
\end{align}

\subsection{Torsion and Curvature Tensors}

\begin{definition}
\label{torsion+curv} Let $\mathbf{u,v\in}\sec TM$. The \textit{torsion\/and
curvature operations} of a connection $\nabla$ are respectively the mappings:
$\mathbf{\tau:}\sec(TM\otimes TM)\rightarrow\sec TM$ and $\mathbf{\rho}%
:\sec(TM\otimes TM)\rightarrow\mathrm{End}(\sec TM)$ given by
\begin{align}
\mathbf{\tau}(\mathbf{u,v})  &  =\nabla_{\mathbf{u}}\mathbf{v}-\nabla
_{\mathbf{v}}\mathbf{u}-[\mathbf{u,v}],\label{top}\\
\mathbf{\rho(u,v)}  &  =\nabla_{\mathbf{u}}\nabla_{\mathbf{v}}-\nabla
_{\mathbf{v}}\nabla_{\mathbf{u}}-\nabla_{\lbrack\mathbf{u,v}]}. \label{cop}%
\end{align}

\end{definition}

\begin{definition}
Let $\mathbf{u,v,w}\in\sec TM$ and $\alpha\in\sec\bigwedge^{1}T^{\ast}M$. The
torsion and curvature tensors of a connection $\nabla$ are the mappings
$\boldsymbol{T}:\sec(\bigwedge^{1}T^{\ast}M\otimes TM\otimes TM)\rightarrow
\mathbb{R}$ and $\boldsymbol{R}:\sec(TM\otimes\bigwedge^{1}T^{\ast}M\otimes
TM\otimes TM)\rightarrow\mathbb{R}$ given by
\end{definition}

\begin{align}
\boldsymbol{T}(\alpha,\mathbf{u,v})  &  =\alpha\left(  \mathbf{\tau
}(\mathbf{u,v})\right)  ,\label{to op}\\
\boldsymbol{R}(\mathbf{w},\alpha,\mathbf{u,v})  &  =\alpha(\mathbf{\rho
(u,v)w}), \label{curv op}%
\end{align}
In an arbitrary basis $\langle\boldsymbol{e}_{\mu}\rangle$ for $T(U\cap V\cap
W)$ and dual basis $\langle\boldsymbol{\theta}^{\mu}\rangle$ for $T^{\ast
}(U\cap V\cap W)$ we have%

\begin{equation}
\boldsymbol{T}(\boldsymbol{\theta}^{\lambda},\boldsymbol{e}_{\mu
},\boldsymbol{e}_{\nu}):=\boldsymbol{T}_{\cdot\mu\nu}^{\lambda\cdot\cdot
}=\boldsymbol{\gamma}_{\cdot\mu\nu}^{\lambda\cdot\cdot}-\boldsymbol{\gamma
}_{\cdot\nu\mu}^{\lambda\cdot\cdot}-\boldsymbol{c}_{\cdot\nu\mu}^{\lambda
\cdot\cdot}\text{ .} \label{TORSIONC}%
\end{equation}%
\begin{align}
\boldsymbol{\mathring{R}}_{\mu\cdot\alpha\beta}^{\cdot\lambda\cdot\cdot}  &
:=\boldsymbol{\mathring{R}}(\boldsymbol{e}_{\mu},\boldsymbol{\theta}^{\lambda
},\boldsymbol{e}_{\alpha},\boldsymbol{e}_{\beta})\nonumber\\
&  =\boldsymbol{e}_{\alpha}(\boldsymbol{\mathring{\gamma}}_{\cdot\beta\mu
}^{\lambda\cdot\cdot})-\boldsymbol{e}_{\beta}(\boldsymbol{\mathring{\gamma}%
}_{\cdot\alpha\mu}^{\lambda\cdot\cdot})-\boldsymbol{\mathring{\gamma}}%
_{\cdot\alpha\kappa}^{\lambda\cdot\cdot}\boldsymbol{\mathring{\gamma}}%
_{\cdot\beta\mu}^{\kappa\cdot\cdot}-\boldsymbol{\mathring{\gamma}}%
_{\cdot\kappa\beta}^{\lambda\cdot\cdot}\boldsymbol{\mathring{\gamma}}%
_{\cdot\alpha\mu}^{\kappa\cdot\cdot}-\boldsymbol{c}_{\cdot\alpha\beta}%
^{\kappa\cdot\cdot}\boldsymbol{\mathring{\gamma}}_{\cdot\kappa\mu}%
^{\lambda\cdot\cdot}\text{ },\nonumber\\
\boldsymbol{R}_{\mu\cdot\alpha\beta}^{\cdot\lambda\cdot\cdot}  &
:=\boldsymbol{R}(\boldsymbol{e}_{\mu},\boldsymbol{\theta}^{\lambda
},\boldsymbol{e}_{\alpha},\boldsymbol{e}_{\beta})\nonumber\\
&  =\boldsymbol{e}_{\alpha}\boldsymbol{(\gamma}_{\cdot\beta\mu}^{\lambda
\cdot\cdot}\boldsymbol{)}-\boldsymbol{e}_{\beta}\boldsymbol{(\gamma}%
_{\cdot\alpha\mu}^{\lambda\cdot\cdot}\boldsymbol{)}-\boldsymbol{\gamma}%
_{\cdot\alpha\kappa}^{\lambda\cdot\cdot}\boldsymbol{\gamma}_{\cdot\beta\mu
}^{\kappa\cdot\cdot}-\boldsymbol{\gamma}_{\cdot\beta\kappa}^{\lambda\cdot
\cdot}\boldsymbol{\gamma}_{\cdot\alpha\mu}^{\kappa\cdot\cdot}-\boldsymbol{c}%
_{\cdot\alpha\beta}^{\kappa\cdot\cdot}\boldsymbol{\gamma}_{\cdot\kappa\mu
}^{\lambda\cdot\cdot}\text{ }. \label{curc}%
\end{align}

\subsubsection{Relation Between the Curvature Tensors of $D$ and $\mathring
{D}$}

The components of the curvature tensors relative to the coordinate basis
associated to the coordinates $\langle x^{\mu}\rangle$ covering $V$ are:
\begin{equation}
R_{\mu\cdot\alpha\beta}^{\cdot\lambda\cdot\cdot}=\mathring{R}_{\mu\cdot
\alpha\beta}^{\cdot\lambda\cdot\cdot}+J_{\mu\cdot\lbrack\alpha\beta]}%
^{\cdot\lambda\cdot\cdot},
\end{equation}
where
\begin{align}
J_{\mu\cdot\alpha\beta\cdot}^{\cdot\lambda\cdot\cdot}  &  :=\mathring
{D}_{\alpha}K_{\beta\mu\cdot}^{\lambda\cdot\cdot}-K_{\cdot\alpha\mu}%
^{\sigma\cdot\cdot}K_{\cdot\beta\sigma}^{\lambda\cdot\cdot},\nonumber\\
J_{\mu\cdot\lbrack\alpha\beta]}^{\cdot\lambda\cdot\cdot}  &  =J_{\mu
\cdot\alpha\beta}^{\cdot\lambda\cdot\cdot}-J_{\mu\cdot\beta\alpha}%
^{\cdot\lambda\cdot\cdot}\text{ }. \label{17}%
\end{align}

We need also the

\begin{proposition}
\label{existthetas} Let $\boldsymbol{Z}\in\sec TV$ be a timelike vector field
such that\texttt{ }$\boldsymbol{g}(\boldsymbol{Z},\boldsymbol{Z})=1$. Then,
there exist, in a coordinate neighborhood $V$, three spacelike vector fields
$\boldsymbol{e}_{\mathbf{i}}$ which together with $\boldsymbol{Z}$ form an
orthogonal moving frame for $x\in V\subset M$ \emph{\cite{choquet}.}
\end{proposition}

\begin{proof}
Suppose that the metric of the manifold in a chart ($\psi,V$) with coordinate
functions $\langle x^{\mu}\rangle$ is \texttt{\ }$\boldsymbol{g}=g_{\mu\nu
}dx^{\mu}\otimes dx^{\nu}$. Let $\boldsymbol{Z}=(Q^{\mu}\partial/\partial
x^{\mu})\in\sec TV$ be an arbitrary reference frame and $\boldsymbol{\alpha
}_{\boldsymbol{Z}}=\boldsymbol{g}(\boldsymbol{Z},$ $)=Z_{\mu}dx^{\mu},$
$Z_{\mu}=g_{\mu\nu}Z^{\nu}$ Then, $g_{\mu\nu}Z^{\mu}Z^{\nu}=1$. Now, define
\begin{gather}
\boldsymbol{\theta}^{0}=(\boldsymbol{\alpha}_{\boldsymbol{Z}})_{\mu}dx^{\mu
}=Z_{\mu}dx^{\mu},\nonumber\\
\gamma_{\mu\nu}=g_{\mu\nu}-Z_{\mu}Z_{\nu}. \label{PRa1}%
\end{gather}
Then the metric $\mathtt{g}$ can be written due to the hyperbolicity of the
manifold as%

\begin{gather}
\boldsymbol{g}=\eta_{\mu\nu}\boldsymbol{\theta}^{\mu}\otimes\boldsymbol{\theta
}^{\mathbf{\nu}},\nonumber\\
\sum\limits_{i=1}^{3}\boldsymbol{\theta}^{i}\otimes\boldsymbol{\theta}%
^{i}=\gamma_{\mu\nu}(x)dx^{\mu}\otimes dx^{\nu}. \label{PRa2}%
\end{gather}
Now, call $\boldsymbol{e}_{0}=\boldsymbol{Z}$ and take $\boldsymbol{e}_{i}$
such that $\boldsymbol{\theta}^{i}(\boldsymbol{e}_{j})=$ $\delta_{j}^{i}$. It
follows immediately that $\boldsymbol{g}(\boldsymbol{e}_{\mu},\boldsymbol{e}%
_{\nu})=\eta_{\mu\nu}$, $\mu,\nu=0,1,2,3\mathbf{.}$
\end{proof}

\section{Observers and Reference Frames}

\begin{definition}
An observer in a Lorentzian structure $\langle M,\boldsymbol{g}\rangle$ is a
timelike curve $\gamma$ pointing to the future such that $\boldsymbol{g}%
(\gamma_{\ast},\gamma_{\ast})=1$.
\end{definition}

\begin{definition}
A reference frame in $U\cap V\cap W\subset M$\ in a Lorentzian structure
$\langle M,\boldsymbol{g}\rangle$ is a timelike vector field $\boldsymbol{Z}$
\emph{($\boldsymbol{g(Z,Z)=}1$) }such that each one of its integral lines is
an observer.
\end{definition}

So, if $\sigma$ is an integral line of $\boldsymbol{Z}$, its parametric
equations are
\begin{equation}
\frac{d(x^{\mu}\circ\sigma(\tau))}{d\tau}=Z^{\mu}(x^{\alpha}(\tau)). \label{9}%
\end{equation}

\begin{definition}
A naturally adapted coordinate system $\langle x^{\mu}\rangle$ to a reference
frame $\boldsymbol{Z}\in\sec TV$ \emph{(}denoted\ $\langle\mathrm{nacs}%
|\boldsymbol{Z}\rangle$\emph{)} is one where the spacelike components of
$\boldsymbol{Z}$ are null. Note that such a chart always exist \emph{\cite{bg}%
}.
\end{definition}

\begin{remark}
The definition of a reference frame in a Lorentzian spacetime or in a
Riemann-Cartan spacetime is the same as above since that definition does not
depends on the additional objects entering these structures.
\end{remark}

\subsection{References Frames in $\langle M,\boldsymbol{g},\mathring
{D},\boldsymbol{\tau}_{\boldsymbol{g}},\uparrow\rangle$ and $\langle
M,\boldsymbol{g},D,\boldsymbol{\tau}_{\boldsymbol{g}},\uparrow\rangle$}

Given a reference frame $\boldsymbol{Z}$ in $U\cap V\cap W\subset M$, consider
the physically equivalent $1$-form field
\begin{equation}
\boldsymbol{\alpha}=\boldsymbol{g}(\boldsymbol{Z},\text{ }). \label{alfa}%
\end{equation}
Then we have:

\begin{equation}
\mathring{D}\boldsymbol{\alpha}=\boldsymbol{\mathring{a}}\otimes
\boldsymbol{\alpha}+\boldsymbol{\mathring{\omega}}+\boldsymbol{\mathring
{\sigma}}+\frac{1}{3}\mathfrak{\mathring{E}}\boldsymbol{h}, \label{PR12}%
\end{equation}
where
\begin{equation}
\boldsymbol{h}=\boldsymbol{g}-\boldsymbol{\alpha\otimes\alpha} \label{h}%
\end{equation}
is the projection tensor, $\boldsymbol{\alpha}$ is the (form)
\textit{acceleration }of $\boldsymbol{Z}$, $\boldsymbol{\mathring{\omega}}$ is
the \textit{rotation}\ tensor \emph{(}or \textit{vortex}\emph{)} of
$\boldsymbol{Z}$, $\boldsymbol{\mathring{\sigma}}$ is the \textit{shear} of
$\boldsymbol{Z}$ and $\mathfrak{\mathring{E}}$ is the \textit{expansion ratio}
of $\boldsymbol{Z}$. In a coordinate chart ($\psi,V$) with coordinate
functions $x^{\mu}$, writing $\boldsymbol{Z}=Z^{\mu}\partial/\partial x^{\mu}$
and $\boldsymbol{h}=(g_{\mu\nu}-Z_{\mu}Z_{\nu})dx^{\mu}\otimes dx^{\nu}$ we
have
\begin{gather}
\boldsymbol{\mathring{a}}=\boldsymbol{g}(\mathring{D}_{_{\boldsymbol{Z}}%
}\boldsymbol{Z},\text{ })=\mathring{D}_{_{\boldsymbol{Z}}}\boldsymbol{\alpha
}\nonumber\\
\mathring{\omega}_{\alpha\beta}=Z_{\left[  \mu;\nu\right]  }h_{\alpha}^{\mu
}h_{\beta}^{\nu},\nonumber\\
\mathring{\sigma}_{\alpha\beta}=[Z_{\left(  \mu;\nu\right)  }-\frac{1}%
{3}\mathfrak{\mathring{E}}h_{\mu\nu}]h_{\alpha}^{\mu}h_{\beta}^{\nu
},\nonumber\\
\mathfrak{\mathring{E}}=\mathring{D}_{\mu}\text{ }Z^{\mu}. \label{PR14}%
\end{gather}

\begin{proof}
The decomposition given by Eq.(\ref{PR12}) can be trivially verified if we use
an orthonormal basis where $\boldsymbol{e}_{0}=\boldsymbol{Z}$, for in this
case $\boldsymbol{\alpha=\theta}^{0}$ and we realize that
\begin{gather}
\mathring{\omega}_{ij}=-\frac{1}{2}\left(  \boldsymbol{\mathring{\gamma}%
}_{\cdot ij}^{0\cdot\cdot}-\boldsymbol{\mathring{\gamma}}_{\cdot ji}%
^{0\cdot\cdot}\right)  =-\frac{1}{2}\boldsymbol{c}_{\cdot ij}^{0\cdot\cdot
},\nonumber\\
\mathbf{\mathring{\sigma}}_{ij}=-\frac{1}{2}\left(  \boldsymbol{\mathring
{\gamma}}_{\cdot ij}^{0\cdot\cdot}+\boldsymbol{\mathring{\gamma}}_{\cdot
ji}^{0\cdot\cdot}\right)  -\frac{1}{3}\mathfrak{\mathring{E}}h_{ij}%
,\nonumber\\
\mathfrak{\mathring{E}=-}\eta^{ij}\boldsymbol{\mathring{\gamma}}_{\cdot
ij}^{0\cdot\cdot}. \label{PR14a}%
\end{gather}

\end{proof}

\begin{remark}
We can show that the vorticity tensor has the same components as the object%
\begin{equation}
g(\star(\boldsymbol{\alpha}\wedge d\boldsymbol{\alpha)},\text{ }),
\label{PR15}%
\end{equation}
where $\star$ is the Hodge star operator. Indeed, we have
\begin{align*}
\star(\boldsymbol{\alpha}\wedge d\boldsymbol{\alpha)}  &  =\star\left(
\boldsymbol{\theta}^{0}\wedge d\boldsymbol{\theta}^{0}\right)  =-\frac{1}%
{2}\boldsymbol{c}_{\cdot ij}^{0\cdot\cdot}\star(\boldsymbol{\theta}^{0}%
\wedge\boldsymbol{\theta}^{i}\wedge\boldsymbol{\theta}^{j})\\
&  =-\boldsymbol{c}_{\cdot23}^{0\cdot\cdot}\boldsymbol{\theta}_{1}%
+\boldsymbol{c}_{\cdot13}^{0\cdot\cdot}\boldsymbol{\theta}_{2}-\boldsymbol{c}%
_{\cdot12}^{0\boldsymbol{\cdot\cdot}}\boldsymbol{\theta}_{3},
\end{align*}
and
\begin{equation}
g(\star(\boldsymbol{\alpha}\wedge d\boldsymbol{\alpha)},\text{ }%
)=\boldsymbol{c}_{\cdot23}^{0\cdot\cdot}\boldsymbol{e}_{1}+\boldsymbol{c}%
_{\cdot31}^{0\cdot\cdot}\boldsymbol{e}_{2}+\boldsymbol{c}_{\cdot12}%
^{0\cdot\cdot}\boldsymbol{e}_{3}=\frac{1}{2}\epsilon^{0ijk}\boldsymbol{c}%
_{\cdot jk}^{0\cdot\cdot}\boldsymbol{e}_{i}. \label{pr16}%
\end{equation}

\end{remark}

\begin{remark}
\emph{Eq.(\ref{PR14}) }is the basis for the classification of reference frames
in a Lorentzian spacetime structure\footnote{For the classification of
reference frames in a Newtonian spacetime structure, see \cite{rqb}.}
\emph{\cite{rr19,rodcap2007,sw}} and in order to be possible to talk about the
classification of reference frames in a Riemann-Cartan spacetime structure we
need the
\end{remark}

\begin{proposition}%
\begin{equation}
D\boldsymbol{\alpha}=\boldsymbol{a}\otimes\boldsymbol{\alpha}%
+\boldsymbol{\omega}+\boldsymbol{\sigma}+\frac{1}{3}\mathfrak{E}%
\boldsymbol{h,} \label{PR17}%
\end{equation}

\begin{subequations}
\begin{gather}
\boldsymbol{a}=D_{\boldsymbol{Z}}\boldsymbol{\alpha,}\label{pr18a}\\
\boldsymbol{\omega=\mathring{\omega}}+\boldsymbol{T}^{0},\text{ \ \ }%
\boldsymbol{\sigma=\mathring{\sigma}}+\frac{1}{3}%
(\boldsymbol{\mathfrak{E-\mathring{E}}})\boldsymbol{\boldsymbol{h}+S}%
^{0},\label{pr18b}\\
\boldsymbol{T}^{0}=\frac{1}{2}\boldsymbol{T}_{\cdot ij}^{0\cdot\cdot
}\boldsymbol{\theta}^{i}\wedge\boldsymbol{\theta}^{j},\text{ \ \ }%
\boldsymbol{S}^{0}=-\frac{1}{2}\boldsymbol{S}_{\cdot ij}^{0\cdot\cdot
}\boldsymbol{\theta}^{i}\otimes\boldsymbol{\theta}^{j}. \label{pr18c}%
\end{gather}

\end{subequations}
\end{proposition}

\begin{proof}
It is a simple exercise using an orthonormal basis where $\boldsymbol{\alpha
=\theta}^{0}$.
\end{proof}

\begin{remark}
We observe that in a Riemann-Cartan spacetime the interpretation of
$\boldsymbol{\omega}$ \emph{(}in the decomposition of $D\boldsymbol{\alpha}$
given by \emph{Eq.(\ref{PR17}))} is the same as $\boldsymbol{\mathring{\omega
}}$ in a Lorentzian spacetime \emph{\cite{rodcap2007}}, i.e., it measures the
rotation that one of the infinitesimally nearby curves to an integral curve
$\gamma$ \emph{(}an\ `observer'\emph{)} of $\boldsymbol{Z}$ had in an
infinitesimal lapse of propertime with relation to an orthonormal basis
Fermi-transported by the `observer' $\gamma$. The interpretation of the terms
$\boldsymbol{\sigma}$ and $\boldsymbol{\mathfrak{E}}$ are also analogous to
the corresponding terms in a Lorentzian spacetime.\ Thus, a reference frame is
non-rotating if $\boldsymbol{\omega}=0$, i.e., $\boldsymbol{\mathring{\omega
}=}-\boldsymbol{T}^{0}$ and \emph{Eq.(\ref{pr18c})} shows that torsion is
indeed related to rotation from the point of view of a Lorentzian spacetime structure.
\end{remark}

\subsection{Inertial Reference Frames in $\langle M\simeq\mathbb{R}%
^{4},\boldsymbol{\eta},\protect\overset{m}{D},\boldsymbol{\tau}%
_{\boldsymbol{\eta}},\uparrow\rangle$}

Now, let $\langle M,\boldsymbol{g},\mathring{D},\boldsymbol{\tau
}_{\boldsymbol{g}},\uparrow\rangle=\langle M\simeq\mathbb{R}^{4}%
,\boldsymbol{\eta},\overset{m}{D},\boldsymbol{\tau}_{\boldsymbol{\eta}%
},\uparrow\rangle$ and let $\langle\mathtt{x}^{\mu}\rangle$ be coordinates in
the Einstein-Lorentz-Poincar\'{e} gauge for $M$. If the matrix with entries
$\eta_{\mu\nu}$ is the diagonal matrix \textrm{diag}$(1,-1,-1,-1)$, we have%
\begin{equation}
\boldsymbol{\eta}=\eta_{\mu\nu}d\mathtt{x}^{\mu}\otimes d\mathtt{x}^{\nu}
\label{mg}%
\end{equation}
If we put\ $\boldsymbol{I}=\partial/\partial\mathtt{x}^{0}$ we see immediately
that that $\langle\mathtt{x}^{\mu}\rangle$ is a $\langle nacs|\boldsymbol{Z}%
\rangle$. We have trivially%
\begin{equation}
\overset{m}{D}\boldsymbol{\alpha}_{\boldsymbol{I}}=0, \label{irfa}%
\end{equation}
which means that for the\ reference frame $\boldsymbol{I}=\partial
\boldsymbol{/}\partial\mathtt{x}^{0}$ we have $\boldsymbol{a}=0$,
$\mathbf{\omega}=0$, $\mathbf{\sigma}=0$, $\mathfrak{E}=0.$

\begin{definition}
A inertial reference frame \emph{(IRF)} in $\langle M\simeq\mathbb{R}%
^{4},\boldsymbol{\eta},\overset{m}{D},\boldsymbol{\tau}_{\boldsymbol{\eta}%
},\uparrow\rangle$ is reference frame $\boldsymbol{I}$ such that
$\overset{m}{D}\boldsymbol{\alpha}_{\boldsymbol{I}}=0$.
\end{definition}

So, inertial reference frames in Special Relativity are not accelerating, not
rotating, have no shear and no deformation. Of course, since $\overset{m}{D}%
_{\boldsymbol{\partial/\partial}\mathtt{x}^{0}}\partial/\partial\mathtt{x}%
^{0}=0$, each one of the integral lines of the vector field $\boldsymbol{I}%
=\partial/\partial\mathtt{x}^{0}$ is a timelike autoparallel (in this case,
also a geodesic) of Minkowski spacetime (a straight line).

\subsection{Is there IRFs in\ Lorentzian and Riemann-Cartan Spacetimes?}

The answer is yes for the Lorentzian case only if we can find a reference
frame $\boldsymbol{I}$ such that $\mathring{D}\boldsymbol{\alpha
}_{\boldsymbol{I}}=0$. In general this equation has no solution in a general
$\langle M,\boldsymbol{g},\mathring{D},\boldsymbol{\tau}_{\boldsymbol{g}%
},\uparrow\rangle$ structure and indeed we have the

\begin{proposition}
\emph{\cite{sw} }An \textbf{IRF} exists in the Lorentzian structure $\langle
M,\boldsymbol{g},\mathring{D},\boldsymbol{\tau}_{\boldsymbol{g}}%
,\uparrow\rangle$ only if the Ricci tensor satisfies%
\begin{equation}
Ricci(\boldsymbol{I},\boldsymbol{Y})=0 \label{irfb}%
\end{equation}
for any $\boldsymbol{Y}\in\sec TM$.
\end{proposition}

\begin{remark}
This excludes, e.g., Friedmann universe spacetimes, Einstein-de Sitter
spacetime. So, no IRF exist in many models of GRT\emph{ }considered to be of
interest by one reason or another by\ `professional relativists'.
\end{remark}

\begin{remark}
The situation in a Riemann-Cartan spacetime is more complicated and will be
analyzed elsewhere, but we observe that in an arbitrary teleparalell spacetime
structure $\langle M,\boldsymbol{g},\overset{e}{\nabla},\boldsymbol{\tau
}_{\boldsymbol{g}},\uparrow\rangle$ the teleparallel basis $\langle
\boldsymbol{e}_{\mu}\rangle$ satisfies $\nabla_{\boldsymbol{e}_{\nu}%
}\boldsymbol{e}_{\mu}=0$. Then the reference frame $\boldsymbol{e}_{0}$ is a
\textbf{IRF}.
\end{remark}

\subsection{Pseudo Inertial Reference Frames}

\begin{definition}
A reference frame $\mathfrak{I}\in\sec TU,$ $U\subset M$ is said to be a
\emph{pseudo inertial} \emph{reference frame (\textbf{PIRF}) i}f
$D_{\mathfrak{I}}\mathfrak{I}=0$, $\boldsymbol{\alpha}_{\mathfrak{I}}\wedge
d\boldsymbol{\alpha}_{\mathfrak{I}}=0$ and $\boldsymbol{\alpha}_{\mathfrak{I}%
}=g(\mathfrak{I},$ $)$.
\end{definition}

This definition means that a \textbf{PIRF} is in free fall and it is non
rotating. It means also that it is at least \emph{locally synchronizable}, but
we are not going to discuss synchronizability here (details may be found,
e.g., in (\cite{rodcap2007}).

\section{What is a \textbf{LIRF} in $\langle M,\boldsymbol{g},\nabla
,\boldsymbol{\tau}_{\boldsymbol{g}},\uparrow\rangle$}

\subsection{Normal Coordinate Functions at $p_{o}\in M$}

In what follows $\nabla$ denotes $\mathring{D}$ or $D$. We will specialize our
discourse at appropriate places.

Let $(\varphi,U)$ be a local chart around $p_{o}\in U$ with coordinate
functions $\langle\xi^{\mu}\rangle$. Let $\gamma:\mathbb{R\supset}I\rightarrow
M$, $\tau\mapsto\gamma(\tau)$ an \emph{autoparallel\footnote{Autoparallels in
a Lorentzian spacetime structure coincide with geodesics.}} in $M$ according
to an arbitrary connection $\nabla$, i.e.,%
\begin{equation}
\nabla_{\gamma_{\ast}}\gamma_{\ast}=0. \label{geodesic}%
\end{equation}
Take arbitrary points $p_{o},q\in M$. Put
\begin{gather}
\gamma(0)=p_{o}\text{, \ \ }\boldsymbol{\xi}_{q}:=\left.  \frac{d}{d\tau
}\right\vert _{p_{o}}=\xi_{q}^{\mu}\left.  \frac{\partial}{\partial\xi^{\mu}%
}\right\vert _{p_{o}}=\xi_{q}^{\mu}\mathbf{e}_{\mu}\in T_{p_{o}}M,\label{10}\\
\xi^{\mu}(p_{o})=\xi_{p_{o}}^{\mu}=0,\text{ \ \ \ }\xi^{\mu}(q)=\xi_{q}^{\mu
}\neq0
\end{gather}

Although the notation looks strange it will become clear in a while. Now, any
autoparallel emanating from $p_{o}$ is specified by a given $\boldsymbol{\xi
}_{q}\in T_{p_{o}}M$. Indeed, take $q$\ `near' $p_{o}$, this statement simply
meaning here that the coordinates difference \ $\triangle\xi^{\mu}=\xi
_{q}^{\mu}-\xi_{p_{o}}^{\mu}=\xi_{q}^{\mu}<<1$. In general there may be many
autoparallels that connect $p_{o}$ to $q$. However, there exists a unique
autoparallel $\gamma_{q}:\mathbb{R\supset}I\rightarrow M$, $\tau\mapsto
\gamma_{q}(\tau)$ such that
\begin{equation}
\gamma_{q}(0)=p_{o},\text{ \ \ }\gamma_{q}(1)=q. \label{11}%
\end{equation}

So, under the above conditions we see that if $\xi_{q}^{\mu}<<1$, \ then $q$
uniquely specifies a vector $\boldsymbol{\xi}_{q}=\xi_{q}^{\mu}\mathbf{e}%
_{\mu}\in T_{p_{o}}M.$ It is evident that $\varphi:q\mapsto\boldsymbol{\xi
}_{q}$ serves as a good coordinate system in a neighbourhood of $p_{o}$. We have

\begin{definition}
$\varphi:q\mapsto(\xi_{q}^{0},\xi_{q}^{1},\xi_{q}^{2},\xi_{q}^{3}):=\{\xi
_{q}^{\mu}\}$ is called \ a normal coordinate chart based on $p_{o}$
\emph{[(nccb}$|p_{o}$\emph{)]} with basis $\mathbf{e}_{\mu}=\left.
\frac{\partial}{\partial\xi^{\mu}}\right\vert _{p_{o}}$.
\end{definition}

Obviously $\varphi(p_{o})=(0,0,0,0).$

\begin{definition}
The so-called exponential map is the mapping%
\begin{gather}
\exp:T_{p_{o}}M\rightarrow M,\text{ \ \ }\exp\boldsymbol{\xi}_{q}%
=q,\nonumber\\
\varphi(\exp\boldsymbol{\xi}_{q})=\{\xi_{q}^{\mu}\}. \label{13}%
\end{gather}

\end{definition}

With respect to the (nccb$|p_{o}$) an autoparallel $\gamma(\tau)$ with
$\gamma(0)=p_{o}$ and $\gamma(1)=q$ is represented by%
\begin{equation}
\varphi(\gamma(\tau))=\{\xi_{q}^{\mu}\tau\}. \label{14}%
\end{equation}

\subsection{Autoparallels Passing Trough $p_{o}$ in $\langle M,\boldsymbol{g}%
,\mathring{D},\boldsymbol{\tau}_{\boldsymbol{g}},\uparrow\rangle$}

If $\gamma(\tau)$ is an autoparallel in the structure $\langle
M,\boldsymbol{g},\mathring{D},\boldsymbol{\tau}_{\boldsymbol{g}}%
,\uparrow\rangle$ it satisfies the equation $\mathring{D}_{\gamma_{\ast}%
}\gamma_{\ast}=0$. Let
\begin{equation}
\mathring{D}_{\partial/\partial\xi^{\mu}}\partial/\partial\xi^{\nu}%
:=\mathring{\Gamma}_{_{\cdot\mu\nu}}^{\alpha\cdot\cdot}\partial/\partial
\xi^{\alpha}. \label{15}%
\end{equation}
We shall prove that the connection coefficients $\mathring{\Gamma}_{_{\cdot
\mu\nu}}^{\alpha\cdot\cdot}$ vanishes at $p_{o}$. Indeed, the coordinate
expression of the autoparallel equation $\mathring{D}_{\gamma_{\ast}}%
\gamma_{\ast}=0$ is%
\begin{gather}
\mathring{D}_{\gamma_{\ast}}\gamma_{\ast}=\mathring{D}_{\frac{d}{d\tau}}%
[\frac{d\xi^{\nu}}{d\tau}\partial/\partial\xi^{\nu}]=\frac{d\xi^{\mu}}{d\tau
}\mathring{D}_{\partial/\partial\xi^{\mu}}\frac{d\xi^{\nu}}{d\tau}%
(\tau)\partial/\partial\xi^{\nu}\nonumber\\
=\frac{d^{2}\xi^{\nu}}{d\tau^{2}}\partial/\partial\xi^{\nu}+\frac{d\xi^{\nu}%
}{d\tau}\frac{d\xi^{\mu}}{d\tau}\mathring{D}_{\partial/\partial\xi^{\mu}%
}\partial/\partial\xi^{\nu}\nonumber\\
=\frac{d^{2}\xi^{\nu}}{d\tau^{2}}\partial/\partial\xi^{\nu}+\frac{d\xi^{\nu}%
}{d\tau}\frac{d\xi^{\mu}}{d\tau}[\mathring{\Gamma}_{_{\cdot\mu\nu}}%
^{\alpha\cdot\cdot}(\varphi(\gamma(\tau))]\partial/\partial\xi^{\alpha
}\nonumber\\
=\left\{  \frac{d^{2}\xi^{\nu}}{d\tau^{2}}+\frac{d\xi^{\alpha}}{d\tau}%
\frac{d\xi^{\mu}}{d\tau}[\mathring{\Gamma}_{_{\cdot\mu\alpha}}^{\nu\cdot\cdot
}(\varphi(\gamma(\tau))]\right\}  \partial/\partial\xi^{\nu}=0. \label{16}%
\end{gather}
Now, since at any point $q^{\prime}$ near $p_{o}$ it is%
\begin{equation}
\xi^{\nu}(\tau)=\xi_{q^{\prime}}^{\nu}\tau\label{geq}%
\end{equation}
we have $\left.  \frac{d^{2}\xi^{\nu}}{d\tau^{2}}\right\vert _{\tau=0}=0$ and
Eq.(\ref{16}) gives immediately in view of the fact that $\mathring{\Gamma
}_{_{\cdot\mu\nu}}^{\alpha\cdot\cdot}=\mathring{\Gamma}_{_{\cdot\nu\mu}%
}^{\alpha\cdot\cdot}$ that
\begin{equation}
\mathring{\Gamma}_{_{\cdot\mu\alpha}}^{\nu\cdot\cdot}(\varphi(\gamma(0))=0.
\label{16b}%
\end{equation}

\begin{remark}
In what follows for simplicity of notation and when no confusion arises\ we
eventually use the sloop notation $\mathring{\Gamma}_{_{\cdot\mu\alpha}}%
^{\nu\cdot\cdot}(p_{o}):=\mathring{\Gamma}_{_{\cdot\mu\alpha}}^{\nu\cdot\cdot
}(\varphi(\gamma(0)).$
\end{remark}

Since it is well known that
\begin{equation}
\mathring{\Gamma}_{_{\cdot\mu\alpha}}^{\nu\cdot\cdot}=\frac{1}{2}g^{\nu\kappa
}\left(  \frac{\partial g_{\kappa\alpha}}{\partial\xi^{\mu}}+\frac{\partial
g_{\kappa\mu}}{\partial\xi^{\alpha}}-\frac{\partial g_{\mu\alpha}}{\partial
\xi^{\kappa}}\right)  ,\label{18}%
\end{equation}
we arrive at the conclusion that at the $p_{o}=\gamma(0)$ we can choose the
normal coordinate functions such that
\begin{equation}
\left.  \boldsymbol{g}\left(  \partial/\partial\xi^{\mu},\partial/\partial
\xi^{\nu}\right)  \right\vert _{p_{o}}=\eta_{\mu\nu},\text{ \ \ }\left.
\frac{\partial g_{\kappa\alpha}}{\partial\xi^{\mu}}\right\vert _{p_{o}%
}=0.\label{19}%
\end{equation}

We can also show through a simple computation that for any $q^{\prime}\in U$,
$q^{\prime}\notin p_{o}$ we have%
\begin{equation}
\left.  \mathring{\Gamma}_{_{\cdot\beta\gamma,\mu}}^{\alpha\cdot\cdot}%
(\xi^{\mu})\right\vert _{q^{\prime}}=\left.  \partial/\partial\xi^{\mu
}\mathring{\Gamma}_{_{\cdot\beta\gamma}}^{\alpha\cdot\cdot}(\xi^{\mu
})\right\vert _{q^{\prime}}=-\frac{1}{3}\left.  (\mathring{R}_{\beta
\cdot\gamma\mu}^{\cdot\alpha\cdot\cdot}(\xi^{\mu})+\mathring{R}_{\gamma
\cdot\beta\mu}^{\cdot\alpha\cdot\cdot}(\xi^{\mu}))\right\vert _{q^{\prime}%
},\label{20}%
\end{equation}
and also that for the chart $(\psi,V)$ with coordinate functions $\langle
x^{\mu}\rangle$ we have for $q^{\prime}$ near $p_{o}$
\begin{align}
\xi^{\mu} &  =x^{\mu}+\frac{1}{2}\mathbf{\mathring{\Gamma}}_{\cdot\alpha\beta
}^{\mu\cdot\cdot}(p_{o})x^{\alpha}x^{\beta},\nonumber\\
x^{\mu} &  =\xi^{\mu}-\frac{1}{2}\boldsymbol{\mathring{\Gamma}}_{\cdot
\alpha\beta}^{\mu\cdot\cdot}(p_{o})\xi^{\alpha}\xi^{\beta},\label{21}%
\end{align}
where $\mathbf{\mathring{\Gamma}}_{\cdot\alpha\beta}^{\mu\cdot\cdot}(p_{o})$
denotes the values of the connection coefficients in the coordinates $\langle
x^{\mu}\rangle$, i.e.,
\begin{equation}
\mathring{D}_{\partial/\partial x^{\mu}}\partial/\partial x^{\nu
}=\mathbf{\mathring{\Gamma}}_{\cdot\mu\nu}^{\alpha\cdot\cdot}\partial/\partial
x^{\alpha}.\label{22}%
\end{equation}

\begin{remark}
Let $\gamma$ $\in U\subset M$ be the world line of an observer in autoparallel
motion in spacetime, i.e., $\mathring{D}_{\gamma_{\ast}}\gamma_{\ast}=0$. Then
the above developments show that we can introduce in $U$ normal coordinate
functions $\langle\xi^{\mu}\rangle$ such that for every $p\in\gamma$ we have
\end{remark}

\begin{gather}
\left.  \frac{\partial}{\partial\xi^{0}}\right\vert _{p\in\gamma}=\gamma
_{\ast|p},\quad\left.  \boldsymbol{g}(\partial/\partial\xi^{\mu}%
,\partial/\partial\xi^{\nu})\right\vert _{p\in\gamma}=\eta_{\mu\nu},\text{
\ \ }\left.  \frac{\partial g_{\kappa\alpha}}{\partial\xi^{\mu}}\right\vert
_{p}=0\nonumber\\
\left.  \mathring{\Gamma}_{\cdot\nu\rho}^{\mu\cdot\cdot}(\xi^{\mu})\right\vert
_{p\in\gamma}=\left.  g^{\mu\alpha}\boldsymbol{g}(\partial/\partial\xi
^{\alpha},\mathring{D}_{\partial/\partial\xi^{\nu}}\partial/\partial\xi^{\rho
})\right\vert _{p\in\gamma}=0.\label{7B}%
\end{gather}

Finally observe that
\begin{equation}
\left.  \mathring{R}_{\mu\cdot\alpha\beta}^{\cdot\text{ }\lambda\cdot\text{
}\cdot\text{ }}\right\vert _{p\in\gamma}=\left.  \frac{\partial\mathring
{\Gamma}_{\cdot\beta\mu}^{\lambda\cdot\cdot}}{\partial\xi^{\alpha}}\right\vert
_{p\in\gamma}-\left.  \frac{\partial\mathring{\Gamma}_{\cdot\alpha\mu
}^{\lambda\cdot\cdot}}{\partial\xi^{\beta}}\right\vert _{p\in\gamma}\label{7C}%
\end{equation}
which is non null if the curvature tensor is non null in $U$.

\subsection{\textbf{LIRF}s in $\langle M,\boldsymbol{g},\mathring
{D},\boldsymbol{\tau}_{\boldsymbol{g}},\uparrow\rangle$}

\begin{definition}
\label{geo} Given a timelike autoparallel line $\gamma\subset U$ $\subset M$
and coordinates\emph{ }$\langle\xi^{\mu}\rangle$ covering $U$\emph{\ }we say
that a reference frame $\boldsymbol{L}=\partial/\partial\xi^{0}\in\sec TU$ is
a \emph{local\ inertial Lorentz reference frame associated to }$\gamma$
\emph{(}\textbf{LIRF}$\gamma$\emph{)\footnote{When no confusion arises and
$\gamma$ is clear from the context we simply write \textbf{LIRF}.}} iff
\begin{gather}
\left.  \boldsymbol{L}\right\vert _{p\in\gamma}=\left.  \frac{\partial
}{\partial\xi^{0}}\right\vert _{p\in\gamma}=\left.  \gamma_{\ast}\right\vert
_{p},\nonumber\\
\left.  \alpha_{\boldsymbol{L}}\wedge d\alpha_{\boldsymbol{L}}\right\vert
_{p\in\gamma}=0,\nonumber\\
\left.  \boldsymbol{g}(\partial/\partial\xi^{\mu},\partial/\partial\xi^{\nu
})\right\vert _{p\in\gamma}=\eta_{\mu\nu},\text{ \ \ }\left.  \frac{\partial
g_{\alpha\beta}}{\partial\xi^{\mu}}\right\vert _{p\in\gamma}=0.\label{7c}%
\end{gather}

\end{definition}

Moreover, we say also that the normal coordinate functions (also called in
Physics textbooks\ local Lorentz coordinate functions) $\langle\xi^{\mu
}\rangle$ are \emph{associated} with the \textbf{LIRF}$\gamma.$

\begin{remark}
It is very important to have in mind that for a \textbf{LIRF}$\gamma$
$\boldsymbol{L}$, in general $\left.  \mathring{D}_{\boldsymbol{L}%
}\boldsymbol{L}\right\vert _{p\notin\gamma}\neq0$ (i.e., only the integral
line $\gamma$ of \textbf{L} is in free fall in general),\textbf{\ }and also
eventually $\left.  \alpha_{\mathbf{L}}\wedge d\alpha_{\mathbf{L}}\right\vert
_{p\notin\gamma}\neq0$, which may be a surprising result. In contrast, a
\emph{PIRF} $\mathfrak{I}$ such that $\left.  \mathfrak{I}\right\vert
_{\gamma}=\left.  \boldsymbol{L}\right\vert _{\gamma}$ has all its integral
lines in free fall and the rotation of the frame is always null in all points
where the frame is defined. Finally its is worth to recall that both
$\mathfrak{I}$ and $\boldsymbol{L}$ may eventually have shear and expansion
even at the points of the autoparallel line $\gamma$ that they have in common.
More details in \emph{\cite{rodcap2007}.}
\end{remark}

Let $\gamma$ be an autoparallel line as in definition \emph{\ref{geo}}. A
section $s$ of the orthogonal frame bundle $FU,U\subset M$ is called an
\emph{inertial moving frame along }$\gamma$ \emph{(}\textbf{IMF}$\gamma
$\emph{)}\ when the set
\begin{equation}
s_{\gamma}=\{(e_{0}(p),e_{1}(p),e_{2}(p),e_{3}(p)),p\in\gamma\cap U\}\subset
s,\label{imf}%
\end{equation}
it such that $\forall p\in\gamma$
\begin{equation}
e_{0}(p)=\left.  \gamma_{\ast}\right\vert _{p\in\gamma},\text{ \ \ }\left.
\mathbf{g}(e_{\mu},e_{\nu})\right\vert _{p\in\gamma}=\eta_{\mu\nu},\text{
\ }\left.  \frac{\partial g_{\alpha\beta}}{\partial\xi^{\mu}}\right\vert
_{p\in\gamma}=0.\label{normal}%
\end{equation}
wich implies
\begin{equation}
\mathring{\Gamma}_{\cdot\nu\rho}^{\mu\cdot\cdot}(p)=g^{\mu\alpha}%
\mathbf{g}(e_{\alpha}(p),\mathring{D}_{e_{\nu}(p)}e_{\rho}(p))=0,\forall
p\in\gamma\text{.}%
\end{equation}

\begin{remark}
\label{EXFU}The existence of $s\in\sec FU$ satisfying the above conditions can
be easily proved. Introduce coordinate functions $<\xi^{\mu}>$ for\emph{\ }$U$
such that at $p_{o}\in\gamma,e_{0}(p_{o})=\left.  \frac{\partial}{\partial
\xi^{0}}\right\vert _{p_{o}}=\gamma_{\ast|p_{o}}$, and $e_{i}(p_{o})=\left.
\frac{\partial}{\partial\xi^{i}}\right\vert _{p_{o}},i=1,2,3$ \emph{(}three
orthonormal vectors\emph{)} satisfying \emph{Eq.(\ref{normal})} and parallel
transport the set $e{_{\mu}(}p_{o}{)}$ along $\gamma$. The set ${e_{\mu}%
(}p_{o}{)}$ will then also be \emph{Fermi} transported since $\gamma$ is a
geodesic and as such they define the standard of \emph{no} rotation along
$\gamma$. See details in \emph{\cite{rodcap2007}}.
\end{remark}

\begin{remark}
Let $\mathfrak{I}\in\sec TV$ be a \textbf{PIRF} and $\gamma$ $\subset U\cap V$
one of its integral lines and let $<\xi^{\mu}>,$ $U\subset M$ be a
\emph{normal coordinate system} through all the points of the world line
$\gamma$ such that $\gamma_{\ast}=\left.  \mathfrak{I}\right\vert _{\gamma}$.
Then, in general $<\xi^{\mu}>$ is not a $(nacs|\mathfrak{I})$ in $U$, i.e.,
$\left.  \mathfrak{I}\right\vert _{p\notin\gamma}\neq\left.  \partial
/\partial\xi^{0}\right\vert _{p\notin\gamma}$ even if $\left.  \mathfrak{I}%
\right\vert _{p\in\gamma}=\left.  \partial/\partial\xi^{0}\right\vert
_{p\in\gamma}$.
\end{remark}

\begin{remark}
It is very much important to recall that a reference frame field as introduced
above is a \emph{mathematical} instrument. It did not necessarily need to have
a material substratum \emph{(}i.e., to be realized as a material physical
system\emph{) }in the points of the spacetime manifold where it is defined.
More properly, we state that the integral lines of the vector field
representing a given reference frame do not need to correspond to worldlines
of real particles. If this crucial aspect is not taken into account we may
incur in serious misunderstandings.
\end{remark}

\begin{remark}
Physics textbooks and even most of the professional articles in GR do not
distinguish between the very different concepts of reference frames,
coordinate systems, sections of the frame bundle and does not leave clear what
is meant by the word \emph{local}. In general what authors mean by a local
inertial reference system is the concept of normal coordinates associated to a
timelike autoparallel curve $\gamma$ as describe above. Moreover, keep in mind
that of course, $\gamma_{\ast}=\frac{d}{d\tau}=\left.  \frac{\partial
}{\partial\xi^{0}}\right\vert _{\gamma}$.
\end{remark}

\subsection{\textbf{LIRF}s in $\langle M,\boldsymbol{g},D,\boldsymbol{\tau
}_{\boldsymbol{g}},\uparrow\rangle$}

We have seen above that we can always introduce around a point $p_{o}\in
U\subset M$ in a Lorentzian $\langle M,\boldsymbol{g},\mathring{D}%
,\boldsymbol{\tau}_{\boldsymbol{g}},\uparrow\rangle$ or in a Riemann-Cartan
$\langle M,\boldsymbol{g},D,\boldsymbol{\tau}_{\boldsymbol{g}},\uparrow
\rangle$ structure a chart $(\varphi,U)$ with normal coordinate functions.

However, it is not licit a priory to assume that the normal coordinate
functions of the two structures coincide. So, we denote by $\langle\zeta^{\mu
}\rangle$ the Riemann-Cartan normal coordinate functions around $p_{o}$ in
what follows. In the case of a Lorentzian structure we found that at \ $p_{o}$
the connection coefficients
\[
\mathring{\Gamma}_{\cdot\mu\nu}^{\alpha\cdot\cdot}(p_{o})=(\mathring
{D}_{\partial/\partial\xi^{\mu}}\partial/\partial\xi^{\nu})\cdot
(g^{\alpha\kappa}\partial/\partial\xi^{\kappa})=0.
\]
However,we are not going to suppose that this is generally the case in a
Riemann-Cartan structure. So, let us investigate which conditions%

\begin{equation}
\boldsymbol{\Gamma}_{\cdot\mu\nu}^{\alpha\cdot\cdot}(p_{o})=(D_{\partial
/\partial\zeta^{\mu}}\partial/\partial\zeta^{\nu})\cdot(g^{\alpha\kappa
}\partial/\partial\zeta^{\kappa}), \label{rccc}%
\end{equation}
must satisfy in normal coordinates $\langle\zeta^{\mu}\rangle$. A
Riemann-Cartan autoparallel $\gamma$ passing through $p_{o}$ and neighboring
points $q^{\prime}$ (in the sense mentioned above), satisfy $D_{\gamma_{\ast}%
}\gamma_{\ast}=0$, and we have
\begin{equation}
\frac{d^{2}\zeta^{\nu}}{d\tau^{2}}+\frac{d\zeta^{\alpha}}{d\tau}\frac
{d\zeta^{\mu}}{d\tau}[\mathbf{\Gamma}_{_{\cdot\mu\alpha}}^{\nu\cdot\cdot
}(\varphi(\gamma(\tau))]=0. \label{geod1}%
\end{equation}

If the autoparallel equation is for points from $p_{o}$ to $q$ given by
$\zeta^{\nu}(\tau)=\zeta_{q^{\prime}}^{\nu}\tau$ (recall Eq.(\ref{geq})) then
since $\left.  \frac{d^{2}\zeta^{\nu}}{d\tau^{2}}\right\vert _{\tau=0}=0$, at
$p_{o}$ we must have $\mathbf{\Gamma}_{_{\cdot(\mu\alpha)}}^{\nu\cdot\cdot
}\mathbf{(}p_{o}\mathbf{)}=\frac{1}{2}\left(  \mathbf{\Gamma}_{_{\cdot
\mu\alpha}}^{\nu\cdot\cdot}(p_{o})+\mathbf{\Gamma}_{_{\cdot\alpha\mu}}%
^{\nu\cdot\cdot}(p_{o})\right)  =0$, i.e.,%

\begin{equation}
\mathbf{\Gamma}_{_{\cdot\mu\alpha}}^{\nu\cdot\cdot}(p_{o})=-\mathbf{\Gamma
}_{_{\cdot\alpha\mu}}^{\nu\cdot\cdot}(p_{o}). \label{cond1}%
\end{equation}

Now, if we recall Eq.(\ref{3a}), Eq.(\ref{3c}), Eq.(\ref{3d}) which gives the
components of the torsion and strain tensors, we see that in the case of
normal coordinates $\langle\zeta^{\mu}\rangle$ we must have
\begin{subequations}
\label{crc}%
\begin{align}
\mathbf{T}_{_{\cdot\mu\alpha}}^{\nu\cdot\cdot}(p_{o})  &  =2\mathbf{\Gamma
}_{_{\cdot\mu\alpha}}^{\nu\cdot\cdot}(p_{o}),\label{a}\\
\mathbf{S}_{_{\cdot\mu\alpha}}^{\nu\cdot\cdot}(p_{o})  &  =-2\mathbf{\mathring
{\Gamma}}_{_{\cdot\mu\alpha}}^{\nu\cdot\cdot}(p_{o}), \label{b}%
\end{align}
which are the conditions that select the normal coordinate functions
$\langle\zeta^{\mu}\rangle$ near $p_{o}$ in a Riemann-Cartan spacetime.
\end{subequations}
\begin{remark}
We did not suppose, of course, that the autoparallels of the Levi-Civita and
Riemann-Cartan connections coincide \emph{(}since this is trivially
false\emph{)\footnote{E.g., the geodesics of the Levi-Civita and the
teleparallel connection on the punctured sphere $\mathring{S}$ are very
different, the latter one are the so-called \ loxodromic spirals and the
former are the maximum circles \cite{rodcap2007}.}}. So, we have the question:
when does the two kinds of autoparallels coincide? If they do coincide then
the Lorentzian and Riemann-Cartan normal coordinate functions around $p_{o}$
must coincide and since for a autoparallel from $p_{o}$ to $q$, it is $\left.
\frac{d^{2}\xi^{\nu}}{d\tau^{2}}\right\vert _{\tau=0}=0$ we must have again
that $\mathbf{\Gamma}_{_{\cdot\mu\alpha}}^{\nu\cdot\cdot}(p_{o}%
)=-\mathbf{\Gamma}_{_{\cdot\alpha\mu}}^{\nu\cdot\cdot}(p_{o})$. But now since
$\mathbf{\mathring{\Gamma}}_{_{\cdot\alpha\beta}}^{\nu\cdot\cdot}(p_{o})=0$ we
arrive at the conclusion that
\begin{subequations}
\label{crca}%
\begin{align}
\mathbf{T}_{_{\cdot\mu\alpha}}^{\nu\cdot\cdot} &  =2\mathbf{\Gamma}%
_{_{\cdot\mu\alpha}}^{\nu\cdot\cdot}(p_{o}),\label{c}\\
\mathbf{S}_{_{\cdot\mu\alpha}}^{\nu\cdot\cdot}(p_{o}) &  =0.\label{d}%
\end{align}
\emph{Eq.(\ref{d})} implies moreover that $\mathbf{T}_{_{\nu\mu\alpha}}%
^{\cdot\cdot\cdot}(p_{o})=-\mathbf{T}_{_{\mu\nu\alpha}}^{\cdot\cdot\cdot
}(p_{o})=-\mathbf{T}_{_{\alpha\mu\nu}}^{\cdot\cdot\cdot}(p_{o})$, i.e., the
torsion tensor must be completely anti-symmetric at all manifold points (since
$p_{o}$ is arbitrary):%
\end{subequations}
\begin{equation}
\mathbf{T}_{_{\mu\alpha\nu}}^{\cdot\cdot\cdot}(p_{o})=\mathbf{T}_{_{[\mu
\alpha\nu]}}^{\cdot\cdot\cdot}(p_{o})\label{u}%
\end{equation}
\emph{Eq.(\ref{u})} is then the condition for the two kinds of autoparallel to
coincide. It is a very particular condition and contrary to what is stated in
\emph{\cite{fabbri1,fabbri2,lang}} it is not satisfied by a general
Riemann-Cartan connection and thus cannot serve the purpose of fixing
coordinate functions that could model \textbf{LIRF} analogous to the ones that
exist in the Lorentzian case\emph{\footnote{Also, \cite{soc} who cites
\cite{fabbri1,fabbri2} did not realize that total antisymmetry of the
components of the torsion tensor is no more than the conditon for two kinds of
autoparallels (the Lorentzians and the Riemann-Cartan ones) to coincide.}}. We
recall moreover that the connection coefficients of the Riemann-Cartan
connection although anti-symmetric using the normal coordinate functions will
be not symmetric if arbitrary coordinate functions $\langle x^{\mu}\rangle$
are use, since we have%
\[
\Gamma_{\cdot\iota\kappa}^{\lambda\cdot\cdot}=\frac{\partial x^{\lambda}%
}{\partial\xi^{\mu}}\frac{\partial\xi^{\rho}}{\partial x^{\iota}}%
\frac{\partial\xi^{\sigma}}{\partial x^{\kappa}}\mathbf{\Gamma}_{\cdot
\rho\sigma}^{\mu\cdot\cdot}+\frac{\partial x^{\lambda}}{\partial\xi^{\mu}%
}\frac{\partial^{2}\xi^{\mu}}{\partial x^{\iota}\partial x^{\kappa}}%
\]
The symmetric part is, of course, the same one that appears also in the
transformation law for the Levi-Civita connection coefficients. We arrive at
the conclusion that only for very particular spacetimes, the ones in which the
strain tensor is null, we can build around a point $p_{o}$ normal coordinate
functions for which \emph{Eqs.(\ref{c})} and \emph{(\ref{d}) }hold and it is
clear that in this case $\left.  \boldsymbol{g}\left(  \partial/\partial
\xi^{\mu},\partial/\partial\xi^{\nu}\right)  \right\vert _{p_{o}}=\eta_{\mu
\nu}$ and $\left.  \frac{\partial g_{\kappa\alpha}}{\partial\xi^{\mu}%
}\right\vert _{p_{o}}=0$. However, for the case of \emph{Eqs.(\ref{a})} and
\emph{(\ref{b})} we cannot have $\left.  \boldsymbol{g}\left(  \partial
/\partial\zeta^{\mu},\partial/\partial\zeta^{\nu}\right)  \right\vert _{p_{o}%
}=\eta_{\mu\nu}$ and $\left.  \frac{\partial g_{\kappa\alpha}}{\partial
\xi^{\mu}}\right\vert _{p_{o}}=0$, for otherwise $\mathring{\Gamma}%
_{_{\cdot\mu\alpha}}^{\nu\cdot\cdot}(p_{o})$ would be null.
\end{remark}

So, in definitive normal coordinate functions are not useful to model a
\textbf{LIRF} in Riemann-Cartan spacetimes. So, what can we do to model such a
\textbf{LIRF} in this case?

To answer that question we need the following result:

\begin{proposition}
\label{rc}Along any timelike autoparallel line $\gamma$ $\subset M$\ in a
Riemann-Cartan spacetime structure\ there exists a section $s$ of the
orthogonal subframe bundle $FU\subset FM,U\subset M$ called an \emph{inertial
moving frame along }$\gamma$ \emph{(}\textbf{IMF}$\gamma$\emph{)}\
\begin{equation}
s_{\gamma}=\{(\boldsymbol{e}_{0}(p),\boldsymbol{e}_{1}(p),\boldsymbol{e}%
_{2}(p),\boldsymbol{e}_{3}(p)),p\in U\subset M\}\subset s,
\end{equation}
such that $\forall p\in\gamma$
\begin{gather}
\left.  \boldsymbol{e}_{0}\right\vert _{p\in\gamma}=\frac{1}{\sqrt{g_{00}}%
}\left.  \gamma_{\ast}\right\vert _{p\in\gamma},\text{ \ \ }\left.
\boldsymbol{g}(\boldsymbol{e}_{\mu},\boldsymbol{e}_{\nu})\right\vert
_{p\in\gamma}=\eta_{\mu\nu},\\
\left.  \lbrack\boldsymbol{e}_{\mu},\boldsymbol{e}_{\nu}]\right\vert
_{p\in\gamma}=\left.  \boldsymbol{c}_{\cdot\mu\nu}^{\beta\cdot\cdot
}\boldsymbol{e}_{\beta}\right\vert _{p\in\gamma},\\
\left.  \boldsymbol{\Gamma}_{\cdot\nu\rho}^{\mu\cdot\cdot}\right\vert
_{p\in\gamma}=0
\end{gather}
and where the $\boldsymbol{c}_{\cdot\nu\mu}^{\alpha\cdot\cdot}$ are not all
null \emph{(}i.e., the $\langle\boldsymbol{e}_{\mu}\rangle$ is not a
coordinate basis\emph{)}.
\end{proposition}

\begin{proof}
Indeed, at that any point $p\in U\subset M$ given $(\varphi,U)$, a
(\textrm{nccb}$\mid p$), with bases $\langle\partial/\partial\zeta^{\mu
}\rangle$ and $\langle d\zeta^{\mu}\rangle$ for $TU$ and $T^{\ast}U$ we can
find\ given an arbitrary vector field $\boldsymbol{X}$, a non coordinate basis
$\langle\boldsymbol{e}_{\mu}\rangle$ and $\langle\boldsymbol{\theta}^{\mu
}\rangle$ for $TU$ and $T^{\ast}U$ by finding a solution $\boldsymbol{\Lambda
}$ to the matrix equation\footnote{$\boldsymbol{X}(\boldsymbol{\Lambda})$ is
the matrix with entries $\boldsymbol{X}(\boldsymbol{\Lambda}_{\nu}^{\mu
})=\mathbf{X}^{\alpha}\partial/\partial\zeta^{\alpha}(\boldsymbol{\Lambda
}_{\nu}^{\mu})$.},%
\begin{equation}
0=\boldsymbol{\Gamma}_{\boldsymbol{X}}^{\prime}=\boldsymbol{\Lambda
}_{\boldsymbol{X}}^{-1}\boldsymbol{\Gamma\Lambda+\Lambda}^{-1}\boldsymbol{X}%
(\boldsymbol{\Lambda}),\label{meq}%
\end{equation}
satisfying the conditions
\begin{gather}
\lbrack\boldsymbol{e}_{\mu},\boldsymbol{e}_{\nu}]=\boldsymbol{c}_{\cdot\mu\nu
}^{\beta\cdot\cdot}\boldsymbol{e}_{\beta},\nonumber\\
\boldsymbol{c}_{\cdot\mu\nu}^{\beta\cdot\cdot}=\Lambda_{\mu}^{\alpha
}\boldsymbol{e}_{\alpha}(\Lambda_{\mu}^{\beta})-\Lambda_{\nu}^{\alpha
}\boldsymbol{e}_{\alpha}(\Lambda_{\mu}^{\beta}).
\end{gather}
where the $\boldsymbol{c}_{\cdot\mu\nu}^{\beta\cdot\cdot}$ are not all null
and where the matrix function $\boldsymbol{\Lambda}$ with entries
$\Lambda_{\mu}^{\nu}$ is defined by
\begin{gather}
D_{\boldsymbol{X}}\boldsymbol{e}_{\nu}:=(\boldsymbol{\Gamma}_{\boldsymbol{X}%
})_{\nu}^{\mu}\boldsymbol{e}_{\mu},\text{ \ \ }D_{\boldsymbol{X}%
}\boldsymbol{e}_{\nu}^{\prime}:=(\boldsymbol{\Gamma}_{\boldsymbol{X}}^{\prime
})_{\nu}^{\mu}\boldsymbol{e}_{\mu}^{\prime},\nonumber\\
\boldsymbol{e}_{\mu}=\Lambda_{\mu}^{\nu}e_{\nu},\text{ \ \ }\boldsymbol{\theta
}^{\mu}=(\Lambda^{-1})_{\nu}^{\mu}\boldsymbol{\theta}^{\nu}.\label{dw}%
\end{gather}
To accomplish our enterprise we choose at an arbitrary $p_{o}\in\gamma$ normal
coordiante functions such that
\[
\boldsymbol{e}_{0}(p_{o})=\frac{1}{\sqrt{g_{00}(p_{o})}}\left.  \gamma_{\ast
}\right\vert _{p_{o}}=\left.  \partial/\partial\zeta^{0}\right\vert _{p_{o}}%
\]
and recall that from Proposition \ref{existthetas} there exists
$\boldsymbol{e}_{i}(p_{o})$, $i=1,2,3$ that together with $\boldsymbol{e}%
_{0}(p_{o})$ satisfy $\left.  \boldsymbol{g}(\boldsymbol{e}_{\mu
},\boldsymbol{e}_{\nu})\right\vert _{p_{o}\in\gamma}=\eta_{\mu\nu}$. So, our
task is simply reduced to find solutions for Eq.(\ref{meq})\footnote{If these
solutions result in a set of non orthonormal frames we get from them an
orthonormal frame by standard procedures.}. Now, taking $\boldsymbol{X}$
$=\partial/\partial\zeta^{\alpha}$ we have that $(\boldsymbol{\Gamma
}_{\boldsymbol{X}})_{\nu}^{\mu}:=\boldsymbol{\Gamma}_{\cdot\alpha\nu}%
^{\mu\cdot\cdot}$ and Eq.(\ref{meq}) is the system of differential equations%
\begin{equation}
\partial/\partial\zeta^{\alpha}\left(  \Lambda_{\kappa}^{\mu}\right)
=-\boldsymbol{\Gamma}_{\cdot\alpha\nu}^{\mu\cdot\cdot}\Lambda_{\kappa}^{\nu
}\label{de}%
\end{equation}
whose solution with given boundary conditions is well known \cite{iliev2}.
Once that solution is known we have that $\left.  D_{\boldsymbol{e}_{\nu}%
}\boldsymbol{e}_{\mu}\right\vert _{p_{0}}=0$ and thus we construct the
\textbf{IMF}$\gamma$ by simply parallel transporting the basis
$\{\boldsymbol{e}_{\mu}(p_{o})\}$ of $T_{p_{o}}M$ along $\gamma$, getting for
any $p\in\gamma$, $\left.  D_{\boldsymbol{e}_{\nu}}\boldsymbol{e}_{\mu
}\right\vert _{p\in\gamma}=0$.$\medskip$
\end{proof}

Taking into account the previous proposition we finally propose the following:

\begin{definition}
\label{LIRFRC}Given a timelike autoparallel line $\gamma\subset U$ $\subset
M$\ in a Riemann-Cartan spacetime structure and coordinate functions\emph{
}$\langle\zeta^{\mu}\rangle$ covering $U\subset M$\emph{\ }we say that a
reference frame $\boldsymbol{L}\in\sec TU$ is a \emph{local\ inertial
reference frame associated to }$\gamma$ \emph{(}\textbf{LIRF}$\gamma$\emph{)}
iff for all $p\in\gamma$ there exists exists a section $s$ of the orthogonal
frame bundle $FU\subset FM,U\subset M$,
\begin{equation}
s_{\gamma}=\{(\boldsymbol{e}_{0}(p),\boldsymbol{e}_{1}(p),\boldsymbol{e}%
_{2}(p),\boldsymbol{e}_{3}(p)),p\in U\subset M\}\subset s,
\end{equation}
where
\begin{gather}
\left.  \boldsymbol{e}_{0}\right\vert _{p\in\gamma}=\left.  \boldsymbol{L}%
\right\vert _{p\in\gamma}=\left.  \frac{\partial}{\partial\zeta^{0}%
}\right\vert _{p\in\gamma}=\left.  \gamma_{\ast}\right\vert _{p\in\gamma
},\nonumber\\
\left.  \boldsymbol{\omega}\right\vert _{p\in\gamma}=\left.  -\boldsymbol{T}%
^{0}\right\vert _{p\in\gamma},\nonumber\\
\left.  \boldsymbol{\Gamma}_{_{\cdot\mu\alpha}}^{\nu\cdot\cdot}\right\vert
_{p\in\gamma}=0,\label{irfrc}%
\end{gather}
Moreover we say that $\langle\zeta^{\mu}\rangle$ are inertial coordinates.
\end{definition}

\begin{remark}
Differently from the case of the \textbf{LIRF} in a Lorentzian spacetime, in a
general Riemann-Cartan spacetime we do not have $\left.  \boldsymbol{g}%
(\partial/\partial\zeta^{\mu},\partial/\partial\zeta^{\nu})\right\vert
_{p\in\gamma}=\eta_{\mu\nu}$ and $\left.  \partial g_{\alpha\beta}%
/\partial\zeta^{\mu}\right\vert _{p\in\gamma}=0$, for otherwise we get
$\left.  \mathbf{\mathring{\Gamma}}_{_{\cdot\mu\alpha}}^{\nu\cdot\cdot
}\right\vert _{p\in\gamma}=0$ which, as we saw above, implies a completely
antisymmetric torsion if we want $\left.  \mathbf{\Gamma}_{\cdot_{\mu\alpha}%
}^{\nu\cdot\cdot}(p)\right\vert _{p\in\gamma}=0$. Moreover, we observe that in
the basis $\{\boldsymbol{e}_{\mu}\}$ the components of the torsion tensor are
according to \emph{Eq.(\ref{TORSIONC})} $\left.  \boldsymbol{T}_{\cdot\mu\nu
}^{\alpha\cdot\cdot}\right\vert _{p\in\gamma}=\left.  -\boldsymbol{c}%
_{\cdot\mu\nu}^{\alpha\cdot\cdot}\right\vert _{p\in\gamma}$ and the components
of the Riemann curvature tensor are $\left.  \boldsymbol{R}_{\mu\cdot
\alpha\beta}^{\cdot\lambda\cdot\cdot}\right\vert _{p\in\gamma}=\left.
\boldsymbol{e}_{\alpha}(\boldsymbol{\gamma}_{\cdot\beta\mu}^{\lambda\cdot
\cdot})\right\vert _{p\in\gamma}-\left.  \boldsymbol{e}_{\beta}%
(\boldsymbol{\gamma}_{\cdot\alpha\mu}^{\lambda\cdot\cdot})\right\vert
_{p\in\gamma}$. Keep in mind that although $\left.  \boldsymbol{\gamma}%
_{\cdot\alpha\mu}^{\lambda\cdot\cdot}\right\vert _{\gamma}=0$, we have that
the Riemann curvature tensor is non null in all $p\in\gamma$ since $\left.
\boldsymbol{e}_{\alpha}(\boldsymbol{\gamma}_{\cdot\beta\mu}^{\lambda\cdot
\cdot})\right\vert _{p\in\gamma}\neq0$.
\end{remark}

\section{Equivalence Principle and Einstein's Most\newline Happy Thought}

At last we want to comment that, as well known, in Einstein's GR one can
easily distinguish (despite some claims on the contrary) in any \textit{real}
\textit{physical laboratory, (}i.e., not one modelled by a timelike worldline)
a true gravitational field from an acceleration field of a given reference
frame in Minkowski spacetime \cite{ohanian,rs2001}. This is because in GR the
\textit{mark} of a real gravitational field is the non null Riemann curvature
tensor of $\mathring{D}$, and the Riemann curvature tensor of $\overset{m}{D}$
(present in the definition of Minkowski spacetime) is null. However if we
interpret a gravitational field as the torsion $2$-forms of an
\emph{effective} teleparallel spacetime\footnote{A teleparallel spacetime is
one equipped with a metric compatible connections for which its Riemann
curvature tensor is null, but its torsion tensor is non null.}
$(M,\boldsymbol{\eta},\overset{e}{\nabla},\boldsymbol{\tau}_{\boldsymbol{\eta
}},\uparrow)$ viewed according to the ideas developed in \cite{fr2010,r2011}
as \emph{deformation} of Minkowski spacetime, then one can also interpret the
acceleration field of an accelerated reference frame in Minkowski spacetime as
generating an effective teleparallel spacetime $(M,\boldsymbol{\eta
},\overset{e}{\nabla},\boldsymbol{\tau}_{\boldsymbol{\eta}},\uparrow)$
structure. This can be done as follows. Let $\boldsymbol{Z}\in\sec TU$,
$U\subset M$ with $\boldsymbol{\eta}(\boldsymbol{Z},\boldsymbol{Z})=1$ an
\textit{accelerated reference frame} on Minkowski spacetime. This means as we
know from Section 3.1 that $\boldsymbol{a}=\overset{m}{D}_{\boldsymbol{Z}%
}\boldsymbol{Z}\neq0$. Put $\boldsymbol{e}_{\mathbf{0}}=\boldsymbol{Z}$ and
define an accelerated reference frame as \textit{non} trivial $\ $if
$\boldsymbol{\theta}^{0}=\boldsymbol{\eta}(\boldsymbol{e}_{0},)$ is not an
exact differential. Next recall that in $V\subset M$ there always exist three
other $\boldsymbol{\eta}$-orthonormal vector fields $\boldsymbol{e}_{i}$,
$i=1,2,3$ such that $\langle\boldsymbol{e}_{\mu}\rangle$ is an
$\boldsymbol{\eta}$-orthonormal \ basis for $TU$, i.e., $\boldsymbol{\eta
}=\eta_{\mathbf{\mu\nu}}\boldsymbol{\theta}^{\mu}\otimes\boldsymbol{\theta
}^{\nu}$, where $\langle\boldsymbol{\theta}^{\mu}\rangle$ is the dual
basis\footnote{In general we will also have that $d\boldsymbol{\theta}^{i}%
\neq0$, $i=1,2,3$.} of $\langle\boldsymbol{e}_{\mu}\rangle$. We then have,
$\overset{m}{D}_{\boldsymbol{e}_{\mathbf{\alpha}}}\boldsymbol{e}_{\beta
}=\boldsymbol{\gamma}_{\cdot\alpha\beta}^{\kappa\cdot\cdot}e_{\kappa}$,
$\overset{m}{D}_{\boldsymbol{e}_{\mathbf{\alpha}}}\boldsymbol{\theta}^{\beta
}=-\boldsymbol{\gamma}_{\cdot\alpha\kappa}^{\beta\cdot\cdot}\boldsymbol{\theta
}^{\kappa}$.

What remains in order to be possible to interpret an acceleration field as a
kind of\ `gravitational field' is to introduce on $M$ a $\eta$-metric
compatible connection $\overset{e}{\nabla}$ such that the $\{\boldsymbol{e}%
_{\mu}\}$ is teleparallel according to it, i.e., $\overset{e}{\nabla
}_{\boldsymbol{e}_{\alpha}}\boldsymbol{e}_{\beta}=0,\overset{e}{\nabla
}_{\boldsymbol{e}_{\alpha}}\boldsymbol{\theta}^{\beta}=0$. Indeed, with this
connection the structure $\langle M\simeq\mathbb{R}^{4},\boldsymbol{\eta
},\overset{e}{\nabla},\boldsymbol{\tau}_{\boldsymbol{\eta}},\uparrow\rangle$
has null Riemann curvature tensor but a non null torsion tensor, whose
components are related\footnote{The explict formulas can be easily derived
using the equations of section 4.5.8 of \cite{rodcap2007} which generalizes
for connections with non null \emph{nonmetricty} tensors Eq.(\ref{3d}).} with
the components of the acceleration $\boldsymbol{a}$ and with the other
coefficients $\boldsymbol{\gamma}_{\cdot\alpha\kappa}^{\beta\cdot\cdot}$ of
the connection $\overset{m}{D}$, which describe the motion on Minkowski
spacetime of a \textit{grid} represented by the orthonormal frame
$\langle\boldsymbol{e}_{\mu}\rangle$. Sch\"{u}cking \cite{schu} thinks that
such a description of the gravitational field makes Einstein most happy
though, i.e., the equivalence principle (understood as equivalence between
acceleration and gravitational field) a legitimate mathematical idea. However,
a \textit{true} gravitational field must satisfy (at least with good
approximation) Einstein equation or the \emph{equivalent} equation for the
tetrad fields $\langle\boldsymbol{e}_{\mu}\rangle$ \cite{fr2010,r2011},
whereas there is no single reason for an acceleration field to satisfy that equation.

\section{Conclusions}

In this paper we have recalled the definitions of observers, reference frames
and naturally adapted coordinate chart\ to a given reference frame. Equipped
with these definitions and some basic results such as the proper meaning of an
inertial reference frame in Minkowski spacetime and the notion of
pseudo-inertial reference frames and locally inertial reference frames in a
Lorentzian spacetime, we showed how to define consistently \emph{locally
inertial reference systems} in a general Riemann-Cartan spacetime structure
$\langle M,\boldsymbol{g},D,\boldsymbol{\tau}_{\boldsymbol{g}},\uparrow
\rangle$. We proved that a set of normal coordinate functions $\langle
\zeta^{\mu}\rangle$ covering a timelike autoparallel do not automatically
define a \textbf{LIRF }in $\langle M,\boldsymbol{g},D,\boldsymbol{\tau
}_{\boldsymbol{g}},\uparrow\rangle$\textbf{ }as it is the case in a Lorentzian
spacetime (recall section 4.2), but the coordinate basis $\langle
\partial/\partial\zeta^{\mu}\rangle$ associated to the normal coordiante
functions $\langle\zeta^{\mu}\rangle$ can be used to define a \textbf{LIRF}%
\ (Definition \ref{LIRFRC}) once we take into account Proposition \ref{rc}. We
also briefly recalled how the concepts of \textbf{LIRF} in Lorentzian and
Riemann-Cartan spacetimes are linked to \textquotedblleft Einstein's most
happy thought\textquotedblright, i.e., the equivalence principle. Summing up
we think that our paper complement and help to clarify presentations of
related issues appearing in excellent papers
\cite{abp2002,hartley,iliev1,iliev2,iliev3,iliev4,nester,vdH}, besides
clarifying some misconceptions like the ones in
\cite{fabbri1,fabbri2,lang,soc, schu} as exposed above.

\end{document}